\documentclass[a4paper,12pt]{article}
\usepackage[top=1 in,bottom=1 in,left=1.0 in,right=1.0 in]{geometry}
\usepackage{amsmath}
\usepackage{amssymb}
\usepackage{amsthm,color,cite}
\usepackage{graphicx,float}
\usepackage[colorlinks,linkcolor=blue,citecolor=blue]{hyperref}
\newtheorem{theorem}{Theorem}

\title{Solutions of  the nonlocal nonlinear Schr\"odinger hierarchy via reduction}

\author{Kui Chen, Da-jun Zhang\footnote{Corresponding author. Email: djzhang@staff.shu.edu.cn}\\
{\small  Department of Mathematics, Shanghai University, Shanghai 200444, P.R. China}
}
\date{\today}

\begin{document}
\maketitle

\begin{abstract}

   In this letter we propose an approach to obtain solutions
   for the nonlocal nonlinear Schr\"{o}dinger hierarchy from the known ones  of
   the Ablowitz-Kaup-Newell-Segur hierarchy by reduction.
   These solutions are presented in terms of double Wronskian and some of them are new.
   The approach is general and can be used for other systems with double Wronskian solutions
   which admit local and nonlocal reductions.

\begin{description}
\item[MSC:] 35Q51, 35Q55
\item[PACS:]
02.30.Ik, 02.30.Ks, 05.45.Yv
\item[Keywords:]
 nonlocal nonlinear Schr\"{o}dinger hierarchy, Ablowitz-Kaup-Newell-Segur hierarchy, double Wronskian, solutions, reduction
\end{description}
\end{abstract}

\section{Introduction}

Recently integrable continuous and semidiscrete nonlocal nonlinear Schr\"{o}dinger (NLS) equations
describing wave propagation in nonlinear PT symmetric media were  found \cite{AM-PRL-2013,AM-PRE-2014}
and have drown much attention. The nonlocal NLS hierarchy are  derived from the
Ablowitz-Kaup-Newell-Segur (AKNS) hierarchy through  the reduction \cite{AM-PRL-2013}
\begin{equation}
r(x,t)=\mp q^*(-x,t)
\label{red}
\end{equation}
where $*$ denotes complex conjugate.
As for solutions, these nonlocal models have been solved through Inverse Scattering Transform, Darboux transformation,
{bilinear approach}, etc. (cf.\cite{AM-Nonl-2016,samra-2014,zhu-2016,zhu-2015,lixutao-2015,sinha-2014,khare-2014,yan-2015,huang2016,Zhou-arxiv}).

In this letter, we propose an approach to construct solutions for the nonlocal NLS hierarchy
directly from the known solutions of the AKNS hierarchy. These solutions are presented in terms of double Wronskian.
Wronskian technique is an elegant way to construct and  verify solutions for soliton equations, which is  developed by
Freeman and Nimmo \cite{nimmo-1983}.
The double Wronskian solutions for the focusing NLS equation was first given in \cite{NF-PLA-NLS}.
In 1990 Liu \cite{liu1990} showed the whole AKNS hierarchy admitted solutions in double Wronskian form.
In this letter we directly apply the reduction \eqref{red} to the solutions of the AKNS hierarchy obtained in \cite{liu1990}
and construct solutions for the nonlocal focusing and defocusing NLS hierarchies.
Some of solutions are new, which correspond to multiple poles of transmission coefficient.

The letter is organized as follows.
We will first in Sec.\ref{sec-2} review known results on the  AKNS hierarchy and their double Wronskian solutions.
Then in Sec.\ref{sec-3} we prove the condition under which
double Wronskian solutions of the even order AKNS hierarchy are reduced to those of the nonlocal NLS hierarchies.
In Sec.\ref{sec-4} we present concrete Wronskian elements and some solutions.
Finally, conclusion is given in Sec.\ref{sec-5}.

\section{AKNS hierarchy and double Wronskian solutions}\label{sec-2}

It is well known that the AKNS spectral problem \cite{AKNS-PRL-1973}
\begin{equation}\label{akns-spectral}
 \left( \begin{array}{c}
 \phi_1 \\
 \phi_2
 \end{array}
 \right)_x
 =\left(
 \begin{array}{cc}
             \lambda & q \\
             r & -\lambda
           \end{array}
 \right)
 \left(   \begin{array}{c}
  \phi_1 \\
  \phi_2
  \end{array}
  \right)
\end{equation}
yields the isospectral  AKNS hierarchy
\begin{equation}\label{akns-hierarchy}
 \left(
   \begin{array}{c}
     q_{t_{n}} \\
     r_{t_{n}} \\
   \end{array}
 \right)=K_{n}=L^n \left(
                      \begin{array}{c}
                        -q \\
                        r \\
                      \end{array}
                    \right),~~n\geq 0,
\end{equation}
where $L$ is a recursion operator
\begin{equation}\label{akns-recusion-operator}
 L= \left(
            \begin{array}{cc}
              -\partial_x+2q \partial_x^{-1}r & 2q \partial_x^{-1}q \\
              -2r \partial_x^{-1}r & \partial_x- 2r \partial_x^{-1}q
            \end{array}
          \right),
\end{equation}
in which $\partial_x=\frac{\partial}{\partial_x}$ and $\partial_x=\partial^{-1}_x=\partial^{-1}_x\partial_x=1$.
An alternative form of \eqref{akns-hierarchy} is
\begin{equation}\label{akns-hierarchy-2}
 \left(
   \begin{array}{c}
     q_{t_{n+1}} \\
     r_{t_{n+1}} \\
   \end{array}
 \right)= L \left(
                      \begin{array}{c}
                        q_{t_{n}} \\
                        r_{t_{n}} \\
                      \end{array}
                    \right),~~ (n\geq1)
\end{equation}
with setting $t_1=x$.

By introducing   rational transformation
\begin{equation}\label{akns-transformation}
 q=2\frac{g}{f},~~r=2\frac{h}{f},
\end{equation}
the hierarchy \eqref{akns-hierarchy-2} can be written as the following bilinear form \cite{newell-1985}
\begin{subequations}\label{akns-bilibear}
\begin{eqnarray}
  &&(2D_{t_{n+1}}-D_{x}D_{t_n})g\cdot f =0,   \\
  &&(2D_{t_{n+1}}-D_{x}D_{t_n})f\cdot h =0 ,   \\
  &&D^2_{x}f\cdot f =8gh,
\end{eqnarray}
\end{subequations}
for $n\geq 1$, where  Hirota's bilinear operator $D$ is defined as\cite{Hirota2004}
\begin{equation*}
 D_x^m D_y^n f(x,y)\cdot g(x,y) = (\partial_x- \partial_{x'})^m(\partial_y- \partial_{y'})^n f(x,y)g(x',y')|_{x'=x,y'=y}.
\end{equation*}
System \eqref{akns-bilibear} admit double Wronskian solutions \cite{liu1990,yin2008}
\begin{subequations}\label{anks-solution-fgh}
 \begin{eqnarray}
 && f= W^{n+1,m+1}(\varphi,\psi)=|\widehat{\varphi}^{(n)}; \widehat{\psi}^{(m)}| , \label{1-solution-of-akns-f}  \\
 && g= W^{n+2,m}(\varphi,\psi)=|\widehat{\varphi}^{(n+1)}; \widehat{\psi}^{(m-1)}| , \label{1-solution-of-akns-g}\\
 && h= W^{n,m+2}(\varphi,\psi)=|\widehat{\varphi}^{(n-1)}; \widehat{\psi}^{(m+1)}|. \label{1-solution-of-akns-h}
 \end{eqnarray}
\end{subequations}
Here,  $\varphi$ and $\psi$ are respectively $(n+m+2)$-th order column vectors
\begin{equation}
\varphi=(\varphi_1,\varphi_2,\cdots,\varphi_{n+m+2})^T,~~
\psi=(\psi_1,\psi_2,\cdots,\psi_{n+m+2})^T
\label{phipsi}
\end{equation}
defined by (with $t_1=x$)
 \begin{equation}\label{akns-var-psi-A-c-d}
  \varphi=\exp{\Bigl(\sum^{\infty}_{j=1} A^j t_j\Bigr)}C, ~~\psi=\exp{\Bigl(- \sum^{\infty}_{j=1} A^j t_j\Bigr)}\overline{C},
 \end{equation}
where $A$ is a $(n+m+2)\times (n+m+2)$ constant matrix
\begin{subequations}\label{akns-A-c-d}
\begin{equation}
A=(k_{ij})_{(n+m+2)\times (n+m+2)},~~k_{ij}\in \mathbb{C}   \label{akns-A},
\end{equation}
and $C, \overline{C}$ are $(n+m+2)$-th order constant column vectors
\begin{equation}
C=(c_1,c_2,\cdots,c_{n+m+2})^T,~~ \overline{C}=(\overline{c}_1,\overline{c}_2,\cdots,\overline{c}_{n+m+2})^T;
\end{equation}
\end{subequations}
$\widehat{\varphi}^{(n)}$ and $\widehat{\psi}^{(m)}$ respectively denote $(n+m+2)\times (n+1)$ and $(n+m+2)\times (m+1)$
Wronski matrices
\begin{subequations}\label{akns-psi-varphi}
  \begin{eqnarray}
  &&\widehat{\varphi}^{(n)} = \Bigl(\varphi,\partial_x \varphi,\partial_x^2\varphi,\cdots,  \partial_x^{n}\varphi\Bigr),\\
  &&\widehat{\psi}^{(m)} = \Bigl(\psi,\partial_x \psi,\partial_x^2\psi,\cdots,\partial_x^{m}\psi\Bigr).
\end{eqnarray}
\end{subequations}

Return to the AKNS hierarchy \eqref{akns-hierarchy}. For a concrete system in this hierarchy,
for example, the second order AKNS system
\begin{subequations} \label{akns-2}
\begin{eqnarray}
&&q_{t_2}= -q_{xx}+2q^2r, \\
&&r_{t_2}=r_{xx}-2r^2q,
\end{eqnarray}
\end{subequations}
one can treat $t_j (j>2)$ as dummy variables which are then absorbed into $C$ and $\overline{C}$, and present their solutions through
\eqref{akns-transformation} and \eqref{anks-solution-fgh} where 
$\varphi$ and $\psi$ are defined as\cite{chen2008}
 \begin{equation}\label{akns-2-var-psi-A-c-d-even}
  \varphi=\exp{( Ax+ A^2 t_2)}C, ~~\psi=\exp{(-Ax- A^2 t_2)}\overline{C}.
 \end{equation}
With the same idea, for the even order AKNS hierarchy
\begin{equation}\label{akns-hierarchy-even}
 \left(
   \begin{array}{c}
     q_{t_{n}} \\
     r_{t_{n}} \\
   \end{array}
 \right)=K_{n}=L^n \left(
                      \begin{array}{c}
                        -q \\
                        r \\
                      \end{array}
                    \right),~~n=2,4,6,\cdots,
\end{equation}
their solutions can be described through \eqref{akns-transformation} and \eqref{anks-solution-fgh} where
\begin{equation}\label{akns-var-psi-A-c-d-even}
  \varphi=\exp{\Bigl( Ax+\sum^{\infty}_{j=1} A^{2j} t_{2j}\Bigr)}C, ~~\psi=\exp{\Bigl(-Ax- \sum^{\infty}_{j=1} A^{2j} t_{2j}\Bigr)}\overline{C}.
 \end{equation}

\section{Reduction to the nonlocal NLS hierarchies}\label{sec-3}

The nonlocal nonlinear Schr\"{o}dinger hierarchies \cite{AM-PRL-2013}
\begin{equation}\label{nls-non-hierarchy}
 iq(x)_{\tau_{2k}} = K^{\pm}_k[q(x),q(-x)],~~k =1,2,3,\ldots
\end{equation}
can be reduced from the even order AKNS hierarchy \eqref{akns-hierarchy-even} with the reductions
\begin{equation}\label{nls-non-reduction-qr}
 r(x)=\pm q^*(-x),~~t_{2k}=i\tau_{2k} ,~~ (x, \tau_{2k} \in \mathbb{R}),
\end{equation}
where the $\pm$ in $K^{\pm}_k$  corresponds to the sign ``$\pm$" in \eqref{nls-non-reduction-qr}.
For the case of $k=1$, the nonlocal focusing NLS equation reads
\begin{equation}\label{nls-non-focusing}
  iq(x)_{\tau_2} = K^{-}_1= q(x)_{xx} + 2q(x)^2q^*(-x),
\end{equation}
and the nonlocal defocusing NLS equation reads
\begin{equation}\label{nls-non-defocusing}
  iq(x)_{\tau_2}  = K^{+}_1 = q(x)_{xx} - 2q(x)^2q^*(-x).
\end{equation}
For $k=2$ we have
\begin{align*}
 iq(x)_{\tau_4} = K^{-}_2= & q(x)_{xxxx}   +8q(x)q(x)_xq^*(-x)+ 6(q(x)_x)^2q^*(-x)   \\
  & +4 q(x)q(x)_x q^*(-x)_x + 2 q(x)^2q^*(-x)_{xx} + 6 q(x)^3q^*(-x)^2 ,
\end{align*}
and
\begin{align*}
  iq(x)_{\tau_4} = K^{+}_2= & q(x)_{xxxx} -8q(x)q(x)_xq^*(-x)-6(q(x)_x)^2q^*(-x)    \\
& -4 q(x)q(x)_x q^*(-x)_x - 2 q(x)^2q^*(-x)_{xx} + 6 q(x)^3q^*(-x)^2.
\end{align*}

In the following, we derive solutions for the nonlocal hierarchies \eqref{nls-non-hierarchy} by imposing the reduction relations \eqref{nls-non-reduction-qr}
on the solutions of the even order AKNS hierarchy \eqref{akns-hierarchy-even}.

\begin{theorem} \label{theorem-nonlocal}
The nonlocal isospectral NLS hierarchies \eqref{nls-non-hierarchy} admit the following solutions
\begin{equation}\label{theorem-q-non}
 q(x)=2\frac{|\widehat{\varphi}^{(n+1)},\widehat{\psi}^{(n-1)}|}{|\widehat{\varphi}^{(n)},\widehat{\psi}^{(n)}|},
\end{equation}
where $\varphi$ and $\psi$ are $(2n+2)$-th order column vectors (i.e. $m=n$ in \eqref{phipsi}),
defined by \eqref{akns-var-psi-A-c-d-even} with $t_{2k}=i\tau_{2k}$ and satisfy
 \begin{subequations}\label{assum-1}
 \begin{align}
 & \psi(x) = T \varphi^*(-x), \label{non-constraint}\\
 & \overline{C}= TC^*,
 \label{DTC}
 \end{align}
 \end{subequations}
in which $T$ is a constant matrix  determined through
\begin{subequations}\label{theorem-non-AT}
\begin{eqnarray}
&& AT=TA^*,~~\label{theorem-1-a-t}\\
&& TT^*=\delta I,~~\delta=\mp 1. \label{theorem-1-t-t}
\end{eqnarray}
\end{subequations}
\end{theorem}

\begin{proof}

We are going to prove that with the assumptions \eqref{assum-1} and \eqref{theorem-non-AT},
$q$ and $r$ defined by \eqref{akns-transformation} with double Wronskians \eqref{anks-solution-fgh}
satisfy the reduction \eqref{nls-non-reduction-qr}.
To achieve that, first, under assumptions of \eqref{DTC} and \eqref{theorem-1-a-t}, it is easy to verify that
\begin{align*}
\psi(x)=&\exp{(-Ax- i\sum^{\infty}_{j=1} A^{2j} \tau_{2j})}\overline{C} \\
       =& \exp{( -(TA^*T^{-1})x-i\sum^{\infty}_{j=1} (TA^*T^{-1})^{2j} \tau_{2j})} TC^*\\
       =& T (\exp{( -Ax+i\sum^{\infty}_{j=1} A^{2j} \tau_{2j})}C)^*\\
       =& T\varphi^*(-x).
\end{align*}
Next, to examine relation between Wronskians, based on \eqref{akns-psi-varphi}, we introduce notation
\begin{align*}
  &\widehat{\varphi}^{(n)}(ax)_{[bx]} = \Bigl(\varphi(ax),\partial_{bx}\varphi(ax),\partial_{bx}^2\varphi(ax),\cdots,\partial_{bx}^{n}\varphi(ax)\Bigr)
\end{align*}
and similar one for $\widehat{\psi}^{(m)}(ax)_{bx}$, where $a,b=\pm 1$.
With such notations, setting $n=m$ and noticing the constraint \eqref{non-constraint}, $f,g,h$ in \eqref{anks-solution-fgh} are expressed as
\begin{subequations}\label{nls-non-proof-2-fgh}
\begin{eqnarray}
&& f(x)=|\widehat{\varphi}^{(n)}(x)_{[x]}; \widehat{\psi}^{(n)}(x)_{[x]}| =|\widehat{\varphi}^{(n)}(x)_{[x]}; T \widehat{\varphi}^{*(n)}(-x)_{[x]}|,  \\
&& g(x)=|\widehat{\varphi}^{(n+1)}(x)_{[x]}; \widehat{\psi}^{(n-1)}(x)_{[x]}|= |\widehat{\varphi}^{(n+1)}(x)_{[x]};T \widehat{\varphi}^{*(n-1)}(-x)_{[x]}|,\\
&& h(x)=|\widehat{\varphi}^{(n-1)}(x)_{[x]}; \widehat{\psi}^{(n+1)}(x)_{[x]}| = |\widehat{\varphi}^{(n-1)}(x)_{[x]}; T \widehat{\varphi}^{*(n+1)}(-x)_{[x]}|.
\end{eqnarray}
\end{subequations}
By calculation we find
\begin{align*}
f^*(-x)&= |\widehat{\varphi}^{(n)}(-x)_{[-x]}; T \widehat{\varphi}^{*(n)}(x)_{[-x]}|^*    \\
& = |T^*|\delta^{n+1}|T \widehat{\varphi}^{*(n)}(-x)_{[-x]}; \widehat{\varphi}^{(n)}(x)_{[-x]}|    \\
& = |T^*|\delta^{n+1}(-1)^{(n+1)^2}| \widehat{\varphi}^{(n)}(x)_{[-x]};T \widehat{\varphi}^{*(n)}(-x)_{[-x]}|   \\
& = |T^*|\delta^{n+1}(-1)^{(n+1)^2}| \widehat{\varphi}^{(n)}(x)_{[x]};T \widehat{\varphi}^{*(n)}(-x)_{[x]}|   \\
& = |T^*|\delta^{n+1}(-1)^{(n+1)^2}f(x),
\end{align*}
where we have made use of $\partial_{-x}=-\partial_x$ and assumed $TT^*=\delta I$ where $\delta$ is a constant and $I$ is the unit matrix.
In a similar way we have
\[g^*(-x)=(-1)^{n(n+2)}|T^*|\delta^{n+2} h(x).\]
With these relations, by making a comparison of $q^*(-x)=2 g^*(-x)/f^*(-x)$ and $r(x)=2 h(x)/f(x)$
one can find that $r(x)= \pm q^*(-x)$ when $\delta=\mp 1$.
Thus we finish the proof.

\end{proof}

We note that to reduce \eqref{akns-hierarchy-even} to solutions for the local NLS hierarchies, 
we need same reduction condition as in Theorem \ref{theorem-nonlocal}
except \eqref{theorem-1-a-t} replaced with $AT+TA^*=0$.

\section{Solutions of the nonlocal NLS hierarchies}\label{sec-4}

In this section, we list Wronskian elements of the solutions  for the nonlocal NLS hierarchies \eqref{nls-non-hierarchy}.

\subsection{Soliton solutions}\label{sec-4-1}

\subsubsection{Focusing case}\label{sec-4-1-1}

For the nonlocal focusing NLS hierarchy
\begin{equation}\label{nls-hie-f}
 iq(x)_{\tau_{2k}} = K^{-}_k[q(x),q(-x)],~~k =1,2,3,\ldots,
\end{equation}
according to Theorem \ref{theorem-nonlocal}, we take $\delta = 1$ and
\begin{subequations}\label{AT-f-sol}
\begin{equation}\label{soli-a}
 A= \left(
      \begin{array}{cc}
        A_1 & 0 \\
        0 & A^*_1 \\      \end{array}
    \right),~~T= \left(
                   \begin{array}{cc}
                     0 & I \\
                     I & 0 \\
                   \end{array}
                 \right),~~C= (c_1,c_2,\cdots,c_{2n+2})^T,
\end{equation}
where
\begin{equation}
A_1 = \mathrm{Diag}(k_1,k_2,\cdots,k_{n+1}),~~ k_j\in \mathbb{C}
\label{A1}
\end{equation}
\end{subequations}
and $k_j$ are distinct.
The Wronskian elements are taken as
\begin{subequations}\label{phi-f-sol}
\begin{equation}\label{soli-fouc-varphi}
 \varphi(x) = \left(
             \begin{array}{c}
               \Phi_1 \\
               \Phi_2 \\
             \end{array}
           \right),
           ~~ \psi(x)= T \varphi^*(-x),
\end{equation}
where
\begin{equation}
\Phi_j=(\varphi_{j,1},\varphi_{j,2},\cdots, \varphi_{j,n+1})^T,~~ (j=1,2),
\label{Phi-j}
\end{equation}
\begin{equation}
\varphi_{1,l}=e^{\xi_l(x,k_l)},~~ \varphi_{2,l}=e^{\xi_l(x,k_l^*)},
\end{equation}
and
\begin{equation}
\xi_l(x,k_l) = k_l x + i \sum_{j=1}^{\infty} k_l^{2j} \tau_{2j}.
\label{xi}
\end{equation}
\end{subequations}

The simplest solution to the nonlocal  focusing NLS equation \eqref{nls-non-focusing}
is obtained by taking $n=0$ in double Wronskian and the solution is given as
\begin{equation}\label{soli-focu-1-soliton}
 q_{\mathrm{1ss}}(x)  = \frac{-4i c_1 c_2\mathrm{Im}(k_1) e^{2k_1 x +i 2k_1^2\tau_2}}
 {|c_1|^2 e^{i 4\mathrm{Im}(k_1)x }e^{-8\mathrm{Re}(k_1)\mathrm{Im}(k_1)\tau_2}- |c_2|^2},
\end{equation}
which is same as the one obtained in \cite{zhu-2015} and \cite{huang2016}, up to some scalar transformations.
And \eqref{soli-focu-1-soliton} is a nonsingular periodic solution of the case of $ \mathrm{Re}(k_1) = 0,~\mathrm{Im}(k_1)\neq 0$ and $|c_1|\neq|c_2| $.
Solution modeled by $k_1$ and $k_2$ is
\[ q_{\mathrm{2ss}}(x)  =2 \frac{E}{F}\]
where
\begin{align*}
E=& c_1c_2(k_1-k_2)e^{2\xi_1+2\xi_2}[c_2^*c_3(k_1-k_1^*)(k_2-k_1^*)e^{2k_1^* x} -c_1^*c_4(k_1-k_2^*)(k_2-k_2^*)e^{2k_2^* x}]\\
                                     &+c_3c_4(k_2^*-k_1^*)e^{2(k_1^*+k_2^*)x}[c_2c_3^*(k_2-k_1^*)(k_2-k_2^*)e^{2\xi_2}
                                     +c_1c_4^*(k_1-k_1^*)(k_2^*-k_1)e^{2\xi_1}],\\
F=& |c_1c_2(k_1-k_2)|^2e^{2\xi_1+2\xi_2}+ |c_3c_4(k_2-k_1)|^2e^{2(k_1^*+k_2^*)x} -|c_1c_4(k_2-k_1^*)|^2 e^{2\xi_1+2k_2^* x}\\
                                     &-|c_2c_3(k_2-k_1^*)|^2e^{2\xi_2+2k_1^* x}
                                     +4\mathrm{Im}(k_1)\mathrm{Im}(k_2)[c_1c_2^*c_3c_4^*e^{2\xi_1+2k_1^* x} +c_1^*c_2c_3^*c_4e^{2\xi_2+2k_2^* x}].
\end{align*}

\subsubsection{Defocusing case}\label{sec-4-1-2}

For the nonlocal defocusing NLS hierarchy
\begin{equation}\label{nls-hie-df}
 iq(x)_{\tau_{2k}} = K^{+}_k[q(x),q(-x)],~~k =1,2,3,\ldots,
\end{equation}
from Theorem \ref{theorem-nonlocal}, we take $\delta =- 1$ and
\begin{equation}\label{soli-a-defocusing}
 A= \left(
      \begin{array}{cc}
        A_1 & 0 \\
        0 & A^*_1 \\      \end{array}
    \right),~~T= \left(
                   \begin{array}{cc}
                     0 & -I \\
                     I & 0 \\
                   \end{array}
                 \right),~~C= (c_1,c_2,\cdots,c_{2n+2})^T,
\end{equation}
where $A_1$ is the diagonal matrix \eqref{A1},
 $\varphi$ and $\psi$ follow the structure \eqref{phi-f-sol} but with $T$ defined as in \eqref{soli-a-defocusing}.

Solutions to  the nonlocal defocusing NLS equation \eqref{nls-non-defocusing} are ($n=0$)
\begin{equation}\label{soli-defocu-1-soliton}
 q_{\mathrm{1ss}}(x)  = \frac{-4i c_1 c_2\mathrm{Im}(k_1) e^{2k_1 x +i 2k_1^2\tau_2}}
 {|c_1|^2 e^{i 4\mathrm{Im}(k_1)x }e^{-8\mathrm{Re}(k_1)\mathrm{Im}(k_1)\tau_2}+ |c_2|^2},
\end{equation}
which is a nonsingular solution of the case of $ \mathrm{Re}(k_1) = 0,~\mathrm{Im}(k_1)\neq 0$ and $|c_1|\neq|c_2| $, and
\[ q_{\mathrm{2ss}}(x)  =2 \frac{G}{H}\]
where
\begin{align*}
G=& c_1c_2(k_1-k_2)e^{2\xi_1+2\xi_2}[c_1^*c_4(k_1-k_2^*)(k_2-k_2^*)e^{2k_2^* x}-c_2^*c_3(k_1-k_1^*)(k_2-k_1^*)e^{2k_1^* x}]\\
                                     &+c_3c_4(k_2^*-k_1^*)e^{2(k_1^*+k_2^*)x}[c_2c_3^*(k_2-k_1^*)(k_2-k_2^*)e^{2\xi_2}
                                     +c_1c_4^*(k_1-k_1^*)(k_2^*-k_1)e^{2\xi_1}], \\
H=& |c_1c_2(k_1-k_2)|^2e^{2\xi_1+2\xi_2}- |c_3c_4(k_2-k_1)|^2e^{2(k_1^*+k_2^*)x}+|c_1c_4(k_2-k_1^*)|^2 e^{2\xi_1+2k_2^* x} \\ 
&-|c_2c_3(k_2-k_1^*)|^2e^{2\xi_2+2k_1^* x} +4\mathrm{Im}(k_1)\mathrm{Im}(k_2)[c_1c_2^*c_3c_4^*e^{2\xi_1+2k_1^* x} -c_1^*c_2c_3^*c_4e^{2\xi_2+2k_2^* x}].
\end{align*}

\subsection{Jordan block solutions}

Jordan block solutions mean the solutions corresponding to the matrix $A$ taking the form
\begin{subequations}
\begin{equation}\label{jordan-a}
 A= \left(
      \begin{array}{cc}
        A_1 & 0 \\
        0 & A^*_1 \\
      \end{array}
    \right)_{(2n+2)\times(2n+2)},
\end{equation}
where $A_1$ is a $(n+1)\times (n+1)$ Jordan block
\begin{equation}
A_1= k_1 I + \sigma,~~ I=(\delta_{j,l})_{(n+1)\times (n+1)},~~ \sigma = (\delta_{j,l+1})_{(n+1)\times (n+1)}.
\end{equation}
\end{subequations}

\subsubsection{Focusing case}\label{sec-4-2-1}

For the nonlocal focusing NLS hierarchy \eqref{nls-hie-f}, to get their solutions, we take
$\delta =1$,
\begin{equation}\label{jordan-foc-T}
  T= \left(
       \begin{array}{cc}
         0 & I \\
         I & 0 \\
       \end{array}
     \right)
\end{equation}
and
\begin{equation}\label{C}
C=(C_1,C_1)^T,~~~ C_1=(1,1,\cdots,1)_{n+1}.
\end{equation}
The Wronskian elements take the form \eqref{soli-fouc-varphi} in which
\begin{subequations}
\begin{eqnarray}
&& \Phi_1(x, k_1)   = \Bigl[ e^{\xi_1(x, k_1)}+\sum^{n}_{l=1}\frac{1}{l!} (\partial_{k_1}^l e^{\xi_1(x,k_1)})\sigma^l \Bigr]  C_1^T, \\
&& \Phi_2(x, k_1^*) =\Bigl[ e^{\xi_1(x, k_1^*)}+\sum^{n}_{l=1}\frac{1}{l!} (\partial_{k_1^*}^l e^{\xi_1(x,k_1^*)})\sigma^l \Bigr] C_1^T,
\end{eqnarray}
\end{subequations}
with
\begin{equation}\label{jordan-scatter}
 \xi_1(x,k_1) = k_1 x + i\sum_{j=1}^{\infty} k_1^{2j} \tau_{2j}.   \\
\end{equation}

For the nonlocal NLS equation \eqref{nls-non-focusing}, when $n=1$, from \eqref{jordan-scatter} we have
\begin{equation} \label{jordan-scatter-case}
 \xi_1(x,k_1) = k_1 x + i k_1^{2}\tau_{2},
\end{equation}
and
\begin{equation}
 \varphi = \left(
             \begin{array}{c}
               e^{\xi_1(x,k_1)}( I + \partial_{k_1} e^{\xi_1(x,k_1)} \sigma)  C_1^T \\
               e^{\xi_1(x,k_1^*)}(I+ \partial_{k_1^*} e^{\xi_1(x,k_1^*)} \sigma) C_1^T \\
             \end{array}
           \right),~~\psi= \left(
                             \begin{array}{c}
                                e^{\xi_1^*(-x,k_1^*)}(I+ \partial_{k_1} e^{\xi_1^*(-x,k_1^*)} \sigma) C_1^T \\
                                e^{\xi_1^*(-x,k_1)}( I + \partial_{k_1^*} e^{\xi_1^*(-x,k_1)} \sigma) C_1^T \\
                             \end{array}
                           \right).
\end{equation}
The corresponding  solution reads
\begin{equation}\label{jordan-fou-solu-case}
 q(x)= \frac{4(k_1^* -k_1)( e^{2\xi_1(x,k_1)}a(k_1,k_1^*)- e^{2\xi_1(x,k_1^*)} a(k_1^*,k_1))}{e^{ 2\xi_1(x,k_1)-2\xi_1(x,k_1^*)}+
 e^{2\xi_1(x,k_1^*)- 2\xi_1(x,k_1)} - b},
\end{equation}
with
\begin{subequations}\label{hha}
\begin{eqnarray}
&& a(k_1,k_1^*) = (k_1^*- k_1) x + 2 i k_1^* (k_1^*-k_1)\tau_2 +1 ,  \\
&& b= 4(k_1-k_1^*)^2 x^2 +  8i(k_1-k_1^*)^2( k_1 + k_1^*) x \tau_2 - 16k_1 k_1^*(k_1-k_1^*)^2 \tau_2^2 +2  .
\end{eqnarray}
\end{subequations}

\subsubsection{Defocusing case}

For the nonlocal defocusing NLS hierarchy \eqref{nls-hie-df},
in Theorem \ref{theorem-nonlocal} we take $\delta = -1$, and $A$, $C$, $\Phi_j$ same as in Sec.\ref{sec-4-2-1},
but
\begin{equation}\label{jordan-defoc-T}
  T= \left(
       \begin{array}{cc}
         0 & -I \\
         I & 0 \\
       \end{array}
     \right).
\end{equation}

For the nonlocal defocusing NLS equation \eqref{nls-non-defocusing}, the simplest solution of Jordan block case is
\begin{equation}\label{jordan-defou-solu-case}
 q(x)= \frac{4(k_1^* -k_1)( e^{2\xi_1(x,k_1)}a(k_1,k_1^*)+ e^{2\xi_1(x,k_1^*)} a(k_1^*,k_1))}{e^{ 2\xi_1(x,k_1)-2\xi_1(x,k_1^*)}+
 e^{2\xi_1(x,k_1^*)- 2\xi_1(x,k_1)} + b},
\end{equation}
where $ a(k_1,k_1^*)$ and $ b$ are denoted by \eqref{hha}.

\section{Conclusions}\label{sec-5}

In this letter by direct reduction  we have derived doubled Wronskian solutions of the nonlocal
NLS hierarchies from the known ones of the even-order AKNS hierarchy.
Compared with the double Wronskian solution obtained in \cite{lixutao-2015} via Darboux transformation
where all $\{k_j\}$ are pure imaginary, here in our solutions all $\{k_j\}$ are complex,
which admits more freedom.
Besides, we also obtain Jordan block solutions which correspond to
multiple-pole case of the transmission coefficient.
Since equation \eqref{theorem-non-AT} admits block diagonal extension,
it is easy to extend our results to case where $A$ is
a block-diagonal composed by diagonal matrix and Jordan blocks.
The reduction procedure developed in this letter is general and can apply to other systems which have double Wronskian solutions
and admit local and nonlocal reductions.

\vskip 15pt
\subsection*{Acknowledgments}
This work was supported by  the NSF of China [grant numbers 11371241, 11631007].

\end{document}